\documentclass[12pt,a4paper]{article}
\usepackage[]{geometry}
\usepackage{amssymb,amsfonts,amsthm,amsmath,dsfont, graphicx}
\usepackage[cp1251]{inputenc}
\usepackage[english]{babel}
\usepackage{mathrsfs}

\usepackage{verbatim}
\usepackage{caption}
\usepackage{subcaption}
\usepackage{listings}
\usepackage{color, multirow}
\usepackage{booktabs}
\usepackage{authblk}

\usepackage[authoryear]{natbib}
\newtheorem{Th}{Theorem}

\newtheorem{Rem}{Remark}
\newtheorem{Cor}{Corollary}

\usepackage{color}

\usepackage{ulem}

\newcommand{\T}{^\prime}

\newcommand{\1}{\boldsymbol{1}}
\newcommand{\w}{\boldsymbol{\omega}}
\newcommand{\bmu}{\boldsymbol{\mu}}

\newcommand{\bSigma}{\boldsymbol{\Sigma}}
\newcommand{\bR}{\boldsymbol{R}}
\newcommand{\br}{\boldsymbol{r}}

\newcommand{\inv}{^{-1}}
\newcommand{\D}{\mathcal{D}}
\newcommand{\Rf}{R_f}
\newcommand{\rf}{r_f}

\begin{document}
	\title{An approximate solution for the power utility optimization under predictable returns}
	\date{}
	\author[]{{\small Dmytro Ivasiuk}}
	\affil[]{{\footnotesize Department of Statistics, European University Viadrina, Frankfurt(Oder), Germany}}
	\affil[]{\footnotesize E-mail: ivasiuk@europa-uni.de}
	\maketitle

	\begin{abstract}		
This work derives an approximate analytical single period solution of the portfolio choice problem for the power utility function. It is possible to do so if we consider that the asset returns follow a multivariate normal distribution. 
It is shown in the literature that the log-normal distribution seems to be a good proxy of the normal distribution in case if the standard deviation of the last one is way smaller than its mean. So we can use this property because this happens to be true for gross portfolio returns. In addition, we present a different solution method that relies on the machine learning algorithm called Gradient Descent. It is a powerful tool to solve a wide range of problems, and it was possible to implement this approach to portfolio selection.
Besides, the paper provides a simulation study, where we compare the derived results with the well-known solution, which uses a Taylor series expansion of the utility function.
	\end{abstract}
	
	
	\section{Introduction}

The current paper discusses portfolio optimization for the power utility investor. The classical problem is to derive the portfolio weights which maximize the expected value for the utility function of wealth \cite[see][]{pennacchi2008theory,brandt2009portfolio}. There are different types of utilities, and each one requires a specific approach for the solution. For instance, one can obtain an exact formula for the quadratic function by estimating the mean and variance of asset returns \cite[see][]{bodnar2015closed}. However, it is hard to find an analytical solution for other utilities because of the difficulty to derive the expected value. For that reason, specific distributional assumptions on asset returns become handy. For example, the normal distribution works well for the exponential utility and the log-normal distribution is helpful for the power utility function \cite[see][]{bodnar2015exact,campbell2002strategic,pennacchi2008theory,bodnar2020mean}. On the other hand, one can apply numerical and approximative results. \cite[see][]{brandt2005simulation,broadie2017numerical}. 

This work focuses on the power utility function, also called constant relative risk aversion (CRRA) utility.  Following John W. Pratt, the CRRA investor is a person whose investment decisions do not rely on initial wealth \cite[see][]{campbell2002strategic,pennacchi2008theory}. In addition, those who follow CRRA preferences would gladly pay a part of their wealth to escape risk. This utility is a popular choice in the literature, and up to now, there are several approaches to find its optimal portfolio. In some papers, authors consider a continuous-time setup with Markov chains for market and Brownian motion for stock prices \cite[see][]{ccanakouglu2010portfolio,ccanakouglu2012hara}. But for discrete-time approximate results are more often because of how the function is defined \cite[see][]{brandt2009portfolio,campbell2002strategic,brandt2005simulation,bodnar2020mean}. 

To derive the optimal weights, we consider that asset returns follow a multivariate normal distribution. Consequently, the resulting portfolio returns will follow a normal distribution as a linear combination of asset returns. However, it gives no benefits in the case of the power utility function. It is convenient to have the log-normal distribution for the gross portfolio returns, to calculate the expected value of a power function of wealth. For that reason, the log-normal distribution comes to mind as an approximation of the normal distribution.
Thus, the paper provides approximate single-period optimal portfolio weights for the power utility function, considering that the asset returns follow a multivariate normal distribution. The approximation is based on a similar behaviour of density functions for the normal  $N(\mu,\sigma^2)$ and log-normal $logN(\ln\mu,\sigma^2/\mu^2)$ distributions if $\sigma/\mu$ approaches zero \cite[see][]{bodnar2020mean}. For instance, if $\mu=1$ and $\sigma=0.04$ then both densities pretty much coincide (see Figure \ref{fig: n ln n}). Indeed, testing simulated samples from the corresponding log-normal distribution would frequently fail to reject the normality hypothesis. Performing Shapiro-Wilk test for $10^5$ samples of 2500 values simulated from the log-normal distribution $logN(0,0.04^2)$ has a $p$-value greater than 0.05 in around 43\% cases. A similar situation holds at a stock market trading: the gross returns fluctuate around value zero within a small interval. Besides, the "$6\sigma$" bound of 24\% would cower 99.9\% of all portfolio returns.  It is also notable that the correct definition of the approximation requires a positive definite $\mu$. That is itself a fair assumption since nobody expects to lose all the invested money. One might also question the normality setup because it implies a possibility of negative wealth. Yet from a practical point of view it is almost impossible ($x\sim N(1,0.1) \Rightarrow  P(x<0)=7.62\cdot10^{-24}$). On the other hand, this issue fully vanishes under the log-normality approximation (because it is the positive definite random variable). But the main reason to use log-normal distribution is that one has an analytical formula to calculate any arithmetic moment, including real-valued.
	
	\begin{figure}[]
		\includegraphics[width=\textwidth]{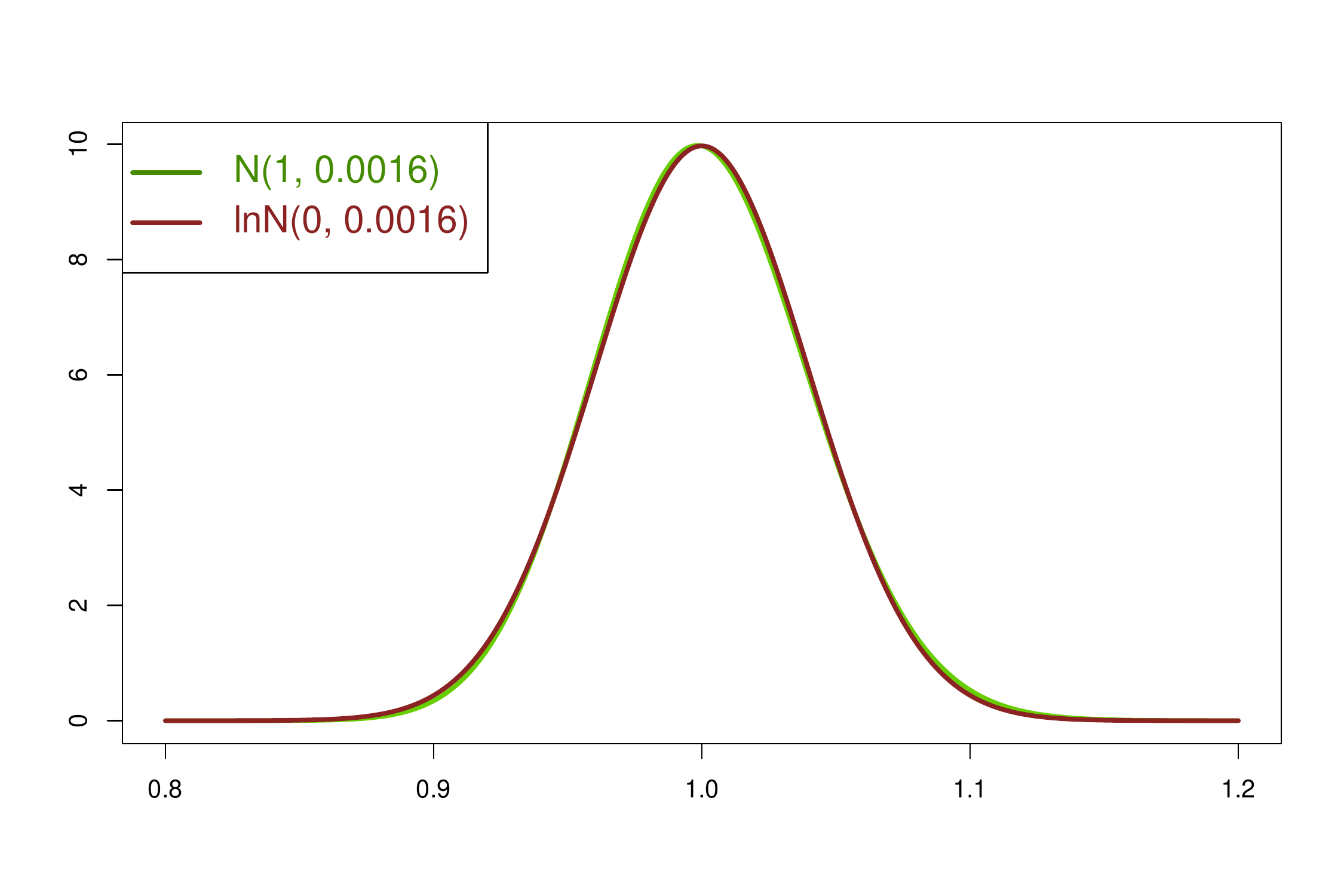}
		\caption{Normal density in case of $\mu_{N}=1$ and $\sigma_{N}=0.04$ and Log-normal density in case of $\mu_{LN}=0$ and $\sigma_{LN}=0.04$}
		\label{fig: n ln n}
	\end{figure}

One should also clarify that the approach presented in this paper is different from \cite{bodnar2020mean}. First of all, we include the risk-free asset into the portfolio. Second, the present work assumes the multivariate normal distribution for the asset returns, which is also used in the literature to explain weekly or monthly returns \cite[see][]{andersen2000exchange}. Whereas in \cite{bodnar2020mean} authors consider the log-normal distribution for gross portfolio returns.

Alongside the analytical solution, we develop the different one which uses a famous machine learning algorithm called Gradient Descend \cite[see][]{friedman2001elements,goodfellow-et-al-2016,geron2019hands}. Generally speaking, this algorithm helps us find the extrema points without solving the first-order conditions, $f^\prime(x)=0$. Besides, there is no need for distributional assumptions on returns because the Gradient Descent "learns" right from the historical data. We also show that this method is applicable in the case of other utilities that are concave.

As an empirical study, we compare the derived results to the benchmark numerical solution from \cite{brandt2005simulation}. To simulate the hypothetical returns, we use a multivariate normal distribution with the mean vector and the variance-covariance matrix estimated from the historical prices of several stocks. It appears that in comparison to the benchmark portfolio, our developed strategies provide the same results, sometimes even being better by a small margin.

	\section{Framework}
	
Below one can find a framework used to present the single-period portfolio selection problem for a power utility function. Let $\boldsymbol{r}=(r_1, r_2,\dots,r_k)\T$ denote the one-period random return vector on $k$ risky assets. Let $\rf$ be the known return (close to zero) on the risk-free asset and $\1$ be the $k$-dimensional vector of ones. Let $\w$ define portfolio weights for a part of an initial wealth $W_0$ allocated between risky assets. We will also consider that the rest of the wealth, $1-\w\1$, stays at a risk-free asset.  Therefore, if $R_{f}=1+r_{f}$ is the gross return on the risk-free asset and $\bR=\br-r_{f}\1$ is the vector of excess return on the risky assets, then the investor's wealth at the end of investment period is $W=W_{0}\left(R_{f}+\w\T\bR\right)$. So it is considered that an investor distributes all the money between risky assets and one risk-free asset without any consumption. Also let us define the mean vector and the variance-covariance matrix of the excess returns with the next way: ${\boldsymbol{\mu}}=E[\br-\rf\1]$, $\boldsymbol{\Sigma}=Var[\br-\rf\1]$.
	
	The investor's problem, analysed in the paper, is to derive the optimal portfolio weights in order to maximise the expected utility of the final wealth:
	\begin{equation}\label{eq:investor's problem}
		\w^*=\underset{\w}{\arg \max} E\left[U\left(W\right)\right],
	\end{equation}
	subject to the budget constraint: $W=W_{0}\left(R_{f}+\w\T\bR\right)$, given the positive initial and final wealth, $W_0$ and $W$ respectively. 

	Precisely speaking, this is a single-period setup of the dynamic utility maximization problem \cite[see,][]{brandt2006dynamic,pennacchi2008theory,bodnar2015closed},
	and the utility function $U(\cdot)$ in our case will be a power function of wealth
	\begin{equation}\label{eq: power utility}
		U(W)=\frac{W^{1-\gamma}}{1-\gamma},\quad\gamma>0, \gamma \neq 1 .
	\end{equation}
	The last one is a constant relative risk aversion (CRRA) utility and belongs to HARA family \cite[see][]{ccanakouglu2010portfolio,ccanakouglu2012hara}.
	
	To model the behaviour of excess returns, we assume the multivariate normal distribution, $N(\bmu,\bSigma)$. Hence, the resulting gross portfolio returns follow the univariate normal distribution, as the linear combination of normal variables, i.e. $$N(\w\T\bmu+\Rf,\w\T\bSigma\w).$$
	 As it was mentioned in the {Introduction} section, we approximate the normal distribution with the log-normal $$logN(\ln(\w\T\bmu+\Rf),\frac{\w\T\bSigma\w^2}{(\w\T\bmu+\Rf)^2})$$ 
	 \cite[see][]{bodnar2020mean}. Figure \ref{fig: n ln n} clearly demonstrates the similarity of both densities. Thus, the approximation should work well because of the small variance of a portfolio.

	Using the approximation, we receive an opportunity to calculate the expected value of a power function of a log-normal random variable:
	\begin{equation}\label{eq: moments of log-normal distribution}
		E\left[X^\lambda\right]=\exp\left(\alpha \lambda+\frac{1}{2}\beta^2\lambda^2\right),\quad \text{if}\ X\sim logN(\alpha,\beta^2),\quad \lambda\in\mathbb{R} .
	\end{equation}
	And base on this property, it was possible to drive an analytical solution for \eqref{eq:investor's problem}.
	
	\section{Approximate solution of the power utility optimization}
	
	The main finding of this paper presents an approximate solution for a single-period portfolio optimization problem \eqref{eq:investor's problem}. We also consider an investor with a power utility of wealth function \eqref{eq: power utility} and the excess returns to follow the multivariate normal distribution.
	
	\begin{Th}\label{th: optimal portfolio weights theorem}
		Assume that the excess returns $\bR$ follow the multivariate normal distribution with the covariance matrix $\bSigma$ and the mean vector $\bmu$. If
			\begin{equation}\label{gamma_min}
				\gamma\geq 1+4\bmu\T\bSigma\inv\bmu,
		\end{equation}
		then the approximate solution of the optimization problem \eqref{eq:investor's problem} with the power utility function \eqref{eq: power utility} is given by
		\begin{equation}\label{eq: optimal portfolio weights}
			\w^*=\Rf\frac{\frac{\gamma - 1}{2} - \bmu\T\bSigma\inv\bmu - \sqrt{\frac{(\gamma - 1)^2}{4} - (\gamma - 1)\bmu\T\bSigma\inv\bmu}}{\left(\bmu\T\bSigma\inv\bmu\right)^2}\bSigma\inv\bmu.
		\end{equation}
	\end{Th}	
	
The proof of the theorem one can find in the \textbf{Appendix} section.
Note that optimal weights \eqref{eq: optimal portfolio weights} do not exist for all the values $\gamma$.   
Next, we will call the result of Theorem \ref{th: optimal portfolio weights theorem} the analytical solution. 
	
\begin{Cor}\label{cor: mean-variance parabola}
	Given by Theorem \ref{th: optimal portfolio weights theorem},
	the corresponding portfolio has its expected excess return and variance as a parabola:
	\begin{equation*}
		\frac{\left(\boldsymbol{\omega}^{*\prime}\bmu\right)^2}{\bmu\T\bSigma\inv\bmu} = {\boldsymbol{\omega}^{*\prime}\bSigma\w}.
	\end{equation*}
\end{Cor}	

\begin{Cor}
	\label{cor: mean and variance decrease of gama}
	Given by Theorem \ref{th: optimal portfolio weights theorem}, the corresponding portfolio has its expected gross return and variance being decrease functions of the coefficient of relative risk aversion $\gamma$.
\end{Cor}
	The proofs of Corollary \ref{cor: mean-variance parabola} and Corollary \ref{cor: mean and variance decrease of gama} see in the \textbf{Appendix} section.
	
Both corollaries describe two economic insights. The first one is similar to the mean-variance efficient frontier trade-off \cite[see,][]{merton1972analytic,pennacchi2008theory,brandt2009portfolio}; namely, the higher expected return is achievable only by taking more risk, and vice versa.
The second one corresponds to the interpretation of risk aversion, which is having a higher $\gamma$ means the less volatile portfolio in exchange for a portion of the profit. 	
\begin{Rem}
	If $\gamma=1+4J$ the optimal weights in Theorem \ref{th: optimal portfolio weights theorem} will be
	$$\w^*_{0}=\Rf\frac{\bSigma\inv\bmu}{\bmu\T\bSigma\inv\bmu},$$
with the resulting portfolio excess return $\bmu\T\w^*_0=\Rf$.
\end{Rem}

\begin{Cor}\label{cor: tangency portfolio}
	Let Theorem \ref{th: optimal portfolio weights theorem} holds. If we consider no investment into the risk-free asset, namely all the wealth distributed between risky assets, then the optimal weights will look as follows:
	$$\w_{tgc}=\frac{\bSigma\inv\bmu}{\1\T\bSigma\inv\bmu}.$$
	The corresponding coefficient of relative risk aversion is given by:
	$$\gamma_{tgc}=\left(\frac{\bmu\T\bSigma\inv\bmu}{\1\T\bSigma\inv\bmu}+\Rf\right)^2\frac{\1\T\bSigma\inv\bmu}{\Rf}+1.$$
\end{Cor}
See the proof of Corollary \ref{cor: tangency portfolio} in the \textbf{Appendix} section.

The weights derived in Corollary \ref{cor: tangency portfolio} represent the single point where the mean-variance efficient frontiers with and without risk-free assets coincide.

\section{Gradient Descent}
Besides the analytical solution, we developed the different one based on Gradient Descent. The last one is a powerful tool used in machine learning and neural networks.  

Gradient Descent is an optimization algorithm that is capable of finding the optimal solution for different problems \cite[see][]{friedman2001elements,goodfellow-et-al-2016, geron2019hands}. This method iteratively reaches the minimum or maximum of the target function. Precisely speaking, it calculates the gradient at the given point and shifts towards the direction of the increasing gradient. The step you reach a zero gradient will guaranty an extremum. But this works only if the target function is concave because one may end up at the local maximum or "flat" area. Initially, we start at a random point and update the solution step by step.  

Let $\{R_i\}_{i=1}^{N}$ be the sample of excess returns simulated from the distribution $\bR$, and let
$$V_0=\frac{1}{N}\sum\limits_{i=1}^{N}\frac{\left(\Rf + \w\T R_i\right)^{1-\gamma}}{1-\gamma}.$$
Then, for a large value $N$, $V_0$ converges the expectation which we need to maximize as the portfolio choice problem:
$$\lim_{N\rightarrow\infty}V_0 = E\left[\frac{\left(\Rf + \w\T\bR\right)^{1-\gamma}}{1-\gamma}\right].$$ 
The Hessian of $V_0$ is negative, 
$$\frac{\partial^2V_0}{\partial\w\partial\w\T} = -\frac{\gamma}{N}\sum\limits_{i=1}^{N}\frac{R_iR_i\T}{\left(\Rf + \w\T R_i\right)^{1+\gamma}}<0,$$
what implies the concavity of $V_0$. Meanwhile, from calculus, we know that a gradient reveals the direction that a function increases. Consequently, the gradient of $V_0$ will always point towards the uphill. Thus, adding the gradient vector to some predefined initial point will create the "path" which leads to the top of the surface.
The process stops when the gradient becomes desirably close to zero. For instance, the euclidean norm can serve as a measure. 

Therefore, one can build the iterative gradient boosting solution given the formula below:
\begin{equation}\label{eq: weights gradient}
	\w_{grd}(i+1)=\w_{grd}(i) + \eta\frac{\partial V_0}{\partial\w}(\w_{grd}(i))=\w_{grd}(i) + \eta\frac{1}{N}\sum\limits_{i=1}^{N}\frac{R_i}{\left(\Rf + \w\T_{grd}(i) R_i\right)^{\gamma}},
\end{equation}
where $\eta$ is a "learning rate" ($\eta=0.1$ is recommended by practitioners). The tuning parameter $\eta$ defines how far we move each iteration. If the steps are too large, the process may not converge by jumping over the solution. On the other hand, very small steps will cause a large number of iteration before reaching the maximum (or minimum for a convex problem)  \cite[see][]{geron2019hands}.

Besides, one can have similar iterative formulas for other utilities, e.g. quadratic and exponential.
Indeed, if $V_0=\frac{1}{N}\sum\limits_{i=1}^{N}U\left(W_0\left(\Rf+\w\T R_i\right)\right)$ and taking into account that the utility is a concave function we have a negative Hessian:

$$\frac{\partial^2V_0}{\partial\w\partial\w\T} =\sum\limits_{i=1}^{N}U^{\prime\prime}\left(W_0\left(\Rf+\w\T R_i\right)\right)W_0^2{R_iR_i\T}<0.$$
Therefore, any concave utility is solvable with this approach.
Moreover, there is no need to know the distribution of returns. This information is automatically coming from historical data.  

	\section{Comparison study}	
In this chapter, we demonstrate the relevance of our findings. The empirical study compares the weights derived in Theorem \ref{th: optimal portfolio weights theorem} and the Gradient Descent solution to the benchmark results introduced by \citet{brandt2005simulation}. The benchmark solution is obtainable by a Taylor series expansion of the utility function, applicable for any given utility. Next, we will call it the numerical solution.

\subsection{Numerical solution}
Presented below, one can find the solution as the fourth-order expansion presented in \citet{brandt2005simulation}.

Let $\bR_{s}$ be the vector of excess returns on the risky assets at a time point $s$ and $\Rf$ be the gross return of a risk-free asset as before. Let $\w_s$ define the portfolio weights for the risky asset chosen at a time point $s$.

Following \citet{brandt2005simulation} we receive:
	\begin{equation}\label{eq: numerical weights for any utility}
		\begin{split}
			&\w_{t}(i+1)\approx-\left\{E_{t}\left[\frac{\partial^2U\left(\hat{W}_T\right)}{\partial\hat{W}_T^2}\prod\limits_{s=t+1}^{T-1}\left(\hat{\w}_s\T\bR_{s+1}+R_f\right)^2\bR_{t+1}\bR_{t+1}\T\right]W_t\right\}^{-1}\\
			&\times\Bigg\{E_{t}\left[\frac{\partial U\left(\hat{W}_T\right)}{\partial\hat{W}_T}\prod\limits_{s=t+1}^{T-1}\left(\hat{\w}_s\T\bR_{s+1}+R_f\right)\bR_{t+1}\right]\\
			&+\frac{1}{2}E_{t}\left[\frac{\partial^3U\left(\hat{W}_T\right)}{\partial\hat{W}_T^3}\prod\limits_{s=t+1}^{T-1}\left(\hat{\w}_s\T\bR_{s+1}+R_f\right)^3\left(\w_{t}\T(i)\bR_{t+1}\right)^2\bR_{t+1}\right]W_t^2\\
			&+\frac{1}{6}E_{t}\left[\frac{\partial^4U\left(\hat{W}_T\right)}{\partial\hat{W}_T^4}\prod\limits_{s=t+1}^{T-1}\left(\hat{\w}_s\T\bR_{s+1}+R_f\right)^4\left(\w_{t}\T(i)\bR_{t+1}\right)^3\bR_{t+1}\right]W_t^3\Bigg\},
		\end{split}
	\end{equation}
	where $\hat{\w}_s$ defines already calculated weights ($s=t+1,\dots,T-1$) and
	\begin{equation}
		\hat{W}_T=W_tR_f\prod\limits_{s=t+1}^{T-1}\left(\hat{\w}_s\T\bR_{s+1}+R_f\right).
	\end{equation}
Note that $E_t[.]$ stands for the expectation conditional on the information available at a time point $t$.

	After calculating the derivatives of the power utility and substituting $\hat{W}_T$, one has: 
	\begin{equation}\label{eq: numerical solution omega_t risk-free}
		\begin{split}
			&\w_{t}(i+1)\approx\frac{1}{\gamma}\left\{E_{t}\left[\prod\limits_{s=t+1}^{T-1}\left(\hat{\w}_s\T\bR_{s+1}+R_f\right)^{1-\gamma}\bR_{t+1}\bR_{t+1}\T\right]\right\}^{-1}\\
			&\times\Bigg\{R_fE_{t}\left[\prod\limits_{s=t+1}^{T-1}\left(\hat{\w}_s\T\bR_{s+1}+R_f\right)^{1-\gamma}\bR_{t+1}\right]\\
			&+\frac{\gamma(\gamma+1)}{2R_f}E_{t}\left[\prod\limits_{s=t+1}^{T-1}\left(\hat{\w}_s\T\bR_{s+1}+R_f\right)^{1-\gamma}\left(\w_{t}\T(i)\bR_{t+1}\right)^2\bR_{t+1}\right]\\
			&-\frac{\gamma(\gamma+1)(\gamma+2)}{6R_f^2}E_{t}\left[\prod\limits_{s=t+1}^{T-1}\left(\hat{\w}_s\T\bR_{s+1}+R_f\right)^{1-\gamma}\left(\w_{t}\T(i)\bR_{t+1}\right)^3\bR_{t+1}\right]\Bigg\},
		\end{split}
	\end{equation}
	where
	\begin{equation}\label{eq: numerical solution omega_t(0) risk-free}
		\begin{split}
			\w_{t}(0)&\approx\frac{1}{\gamma}\left\{E_{t}\left[\prod\limits_{s=t+1}^{T-1}\left(\hat{\w}_s\T\bR_{s+1}+R_f\right)^{1-\gamma}\bR_{t+1}\bR_{t+1}\T\right]\right\}^{-1}\\
			&\times R_f E_{t}\left[\prod\limits_{s=t+1}^{T-1}\left(\hat{\w}_s\T\bR_{s+1}+R_f\right)^{1-\gamma}\bR_{t+1}\right]
		\end{split}
	\end{equation}
	and
	\begin{equation}\label{eq: numerical solution omega_T risk-free}
		\begin{split}
			&\w_{T-1}(i+1)\approx\frac{1}{\gamma}\left\{E_{T-1}\left[\bR_{T}\bR_{T}\T\right]\right\}^{-1}\\
			&\times\Bigg\{R_fE_{T-1}\left[\bR_{T}\right]+\frac{\gamma(\gamma+1)}{2R_f}E_{T-1}\left[\left(\w_{T-1}\T(i)\bR_{T}\right)^2\bR_{T}\right]\\
			&-\frac{\gamma(\gamma+1)(\gamma+2)}{6R_f^2}E_{T-1}\left[\left(\w_{T-1}\T(i)\bR_{T}\right)^3\bR_{T}\right]\Bigg\},
		\end{split}
	\end{equation}
	with
	\begin{equation}\label{eq: numerical solution omega_T(0) risk-free}
		\begin{split}
			&\w_{T-1}(0)\approx\frac{1}{\gamma}\left\{E_{T-1}\left[\bR_{T}\bR_{T}\T\right]\right\}^{-1}E_{T-1}\left[\bR_{T}\right].
		\end{split}
	\end{equation}
	Hence, this is a complete backward scheme for the numerical solution of a multi-period investment strategy. As soon as we are in a single-period investment framework, formulas \eqref{eq: numerical solution omega_T risk-free} and \eqref{eq: numerical solution omega_T(0) risk-free} produce the desired weights.
	
	Thus,
	\begin{equation*}
		\begin{split}
			&\w_{num}(i+1)\approx\frac{1}{\gamma}\left\{E\left[\bR\bR\T\right]\right\}^{-1}\\
			&\times\left\{R_fE\left[\bR\right]+\frac{\gamma(\gamma+1)}{2R_f}E\left[\left(\w\T(i)\bR\right)^2\bR\right] - \frac{\gamma(\gamma+1)(\gamma+2)}{6R_f^2}E\left[\left(\w\T(i)\bR\right)^3\bR\right]\right\},
		\end{split}
	\end{equation*}
	with
	\begin{equation*}
		\begin{split}
			&\w_{num}(0)\approx\frac{1}{\gamma}\left\{E\left[\bR\bR\T\right]\right\}^{-1}E\left[\bR\right].
		\end{split}
	\end{equation*}
	And given the normality of returns, $\bR\sim N(\bmu,\bSigma)$, we have
	
		\begin{equation}\label{eq: numerical solution omega risk-free}
		\begin{split}
			&\w_{num}(i+1)\approx\frac{1}{\gamma}\left\{\bSigma+\bmu\bmu\T\right\}^{-1}\\
			&\times\left\{R_f\bmu + \frac{\gamma(\gamma + 1)}{2R_f}E\left[\left(\w\T(i)\bR\right)^2\bR\right] - \frac{\gamma(\gamma+1)(\gamma+2)}{6R_f^2}E\left[\left(\w\T(i)\bR\right)^3\bR\right]\right\},
		\end{split}
	\end{equation}
	with
	\begin{equation}\label{eq: numerical solution omega(0) risk-free}
		\begin{split}
			&\w_{num}(0)\approx\frac{\Rf}{\gamma}\left\{\bSigma+\bmu\bmu\T\right\}^{-1}\bmu.
		\end{split}
	\end{equation}
	
	The only task left is to calculate the expectations. Similarly as in Gradient Descent method, the law of large numbers becomes handy. If $f(x)$ defines a probability density function of a random variable $x$, then the expectation $E\left[g(x)\right]$ can be estimated by 
	\begin{equation*}
		\frac{1}{N}\sum\limits_{i=1}^Ng(x_i)
	\end{equation*}
	where $\{x_i\}_{i=1}^N$ is a large sample taken from the distribution given by $f(x)$ \cite[see][]{barberis2000investing}.

	 	\subsection{Portfolio comparison}
	 	
	 	In the comparison study, we simulate the large number ($N=10^6$) of hypothetical vectors of asset returns given the multivariate normal distribution. Using this data set, we calculate the weights for numerical solution and the Gradient Descent. At the same time, we calculate the optimal weights for the analytical solution using the mean vector and the covariance matrix from the distribution. This is repeated for several values of the risk aversion parameter, i.e. $\gamma\in\{5,10,15,20\}$. As a result, we receive a set of three weight vectors for each $\gamma$.
	 	Next, we calculate the resulting portfolio returns and final utilities using the simulated data and the weights.
	 	Without the loss of generality, one can consider that the initial wealth $W_0$ is equal to one. 
	 	
	We estimate the sample mean vector and the covariance matrix from the stock data to obtain the parameters of the multivariate normal distribution. We took weekly prices from the 16th of January 2012 till the 9th of September 2019 for the DAX index, Nasdaq Futures and Russell 2000 Futures.

	The resulting mean vector and covariance matrix 
	\begin{equation*}
		\bmu=\begin{bmatrix}
			0.00134-\rf \\
			0.00231-\rf \\
			0.00139-\rf
		\end{bmatrix},
		\bSigma=
	\begin{bmatrix}
		 0.000545 && 0.000319 && 0.000341\\
		 0.000319 && 0.000410 && 0.000393\\
		 0.000341 && 0.000393 && 0.000487\\
	\end{bmatrix},
	\end{equation*}
	stay for parameters of the multivariate normal distribution $N(\bmu,\bSigma)$ for hypothetical excess returns. It is worth mentioning that we consider the risk-free rate $\rf$ of 0.06\%, resulting in 3\% of yearly yield. 
	
	\begin{table}[]
			\begin{center}
		\begin{tabular}{rrrrr}
			\hline
			$\gamma$& & &\multicolumn{1}{r}{}&  \\ 
			\cline{2-5} 
			& \multicolumn{1}{r}{Analytical} & \multicolumn{1}{r}{Numerical} & \multicolumn{1}{r}{Gradient Descent} &  \multicolumn{1}{r}{Measure}\\
			\hline 
			5  & 
			\begin{tabular}[c]{@{}l@{}} -0.24761 \\ 0.03487 \\ {-0.24461} \\ 0.03387 
			\end{tabular}   & 
			
			\begin{tabular}[c]{@{}l@{}} -0.24748 \\ 0.02747 \\ -0.24560 \\ 0.02698				
			\end{tabular}    & 
			
			\begin{tabular}[c]{@{}l@{}} -0.24748 \\ 0.02699 \\ -0.24566 \\ 0.02653
			\end{tabular}       & 
			
			\begin{tabular}[c]{@{}l@{}}sample mean\\ sample standard deviation \\ median \\ median absolute deviation\end{tabular} \\
			\cline{2-5}
			10  & 
			\begin{tabular}[c]{@{}l@{}} -0.10957 \\ 0.01530 \\ {-0.10840} \\ 0.01497
			\end{tabular}   & 
			
			\begin{tabular}[c]{@{}l@{}} -0.10956 \\ 0.01369 \\ -0.10861 \\ 0.01345
			\end{tabular}    & 
			
			\begin{tabular}[c]{@{}l@{}} -0.10956 \\ 0.01343 \\ -0.10865 \\ 0.01321
			\end{tabular}       & 
			
		\begin{tabular}[c]{@{}l@{}}sample mean\\ sample standard deviation \\ median \\ median absolute deviation\end{tabular} \\
			\cline{2-5} 
			15 & 
			\begin{tabular}[c]{@{}l@{}} -0.07020 \\ 0.00978 \\ {-0.06947} \\ 0.00959
			\end{tabular}  & 
			
			\begin{tabular}[c]{@{}l@{}} -0.07020 \\ 0.00909 \\ -0.06957 \\ 0.00894
			\end{tabular}   & 
			
			\begin{tabular}[c]{@{}l@{}} -0.07020 \\ 0.00892 \\ -0.06959 \\ 0.00878
			\end{tabular}    & 
			
			\begin{tabular}[c]{@{}l@{}}sample mean\\ sample standard deviation \\ median \\ median absolute deviation\end{tabular}\\
			\cline{2-5} 
			20 & 
			\begin{tabular}[c]{@{}l@{}} -0.05156 \\ 0.00718 \\ {-0.05104} \\ 0.00704
			\end{tabular} & 
			
			\begin{tabular}[c]{@{}l@{}} -0.05156 \\ 0.00680 \\ -0.05109 \\ 0.00668
			\end{tabular} & 
			
			\begin{tabular}[c]{@{}l@{}} -0.05156 \\ 0.00667 \\ -0.05111 \\ 0.00657
			\end{tabular} & 			
			\begin{tabular}[c]{@{}l@{}}sample mean\\ sample standard deviation \\ median \\ median absolute deviation\end{tabular}\\
			\hline
		\end{tabular}
	\end{center}
		\caption{The sample means of the power utility of the final wealth for the numerical solution and the derived strategy for $\gamma \in \{5, 10, 15, 20\}$}
		\label{table: expected power utility}
	\end{table}
	
	\begin{figure}[]		
		\includegraphics[width=0.5\textwidth]{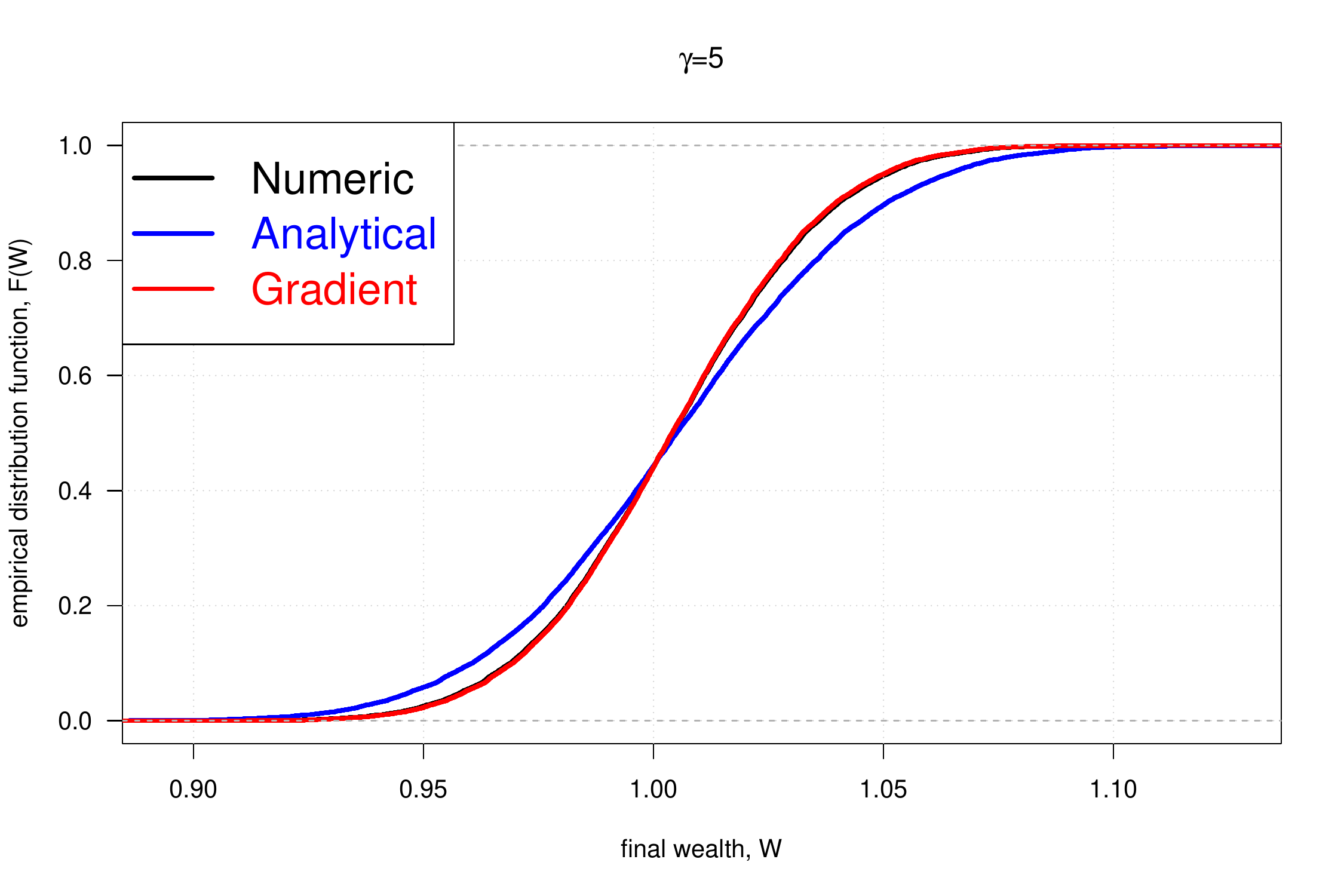}\includegraphics[width=0.5\textwidth]{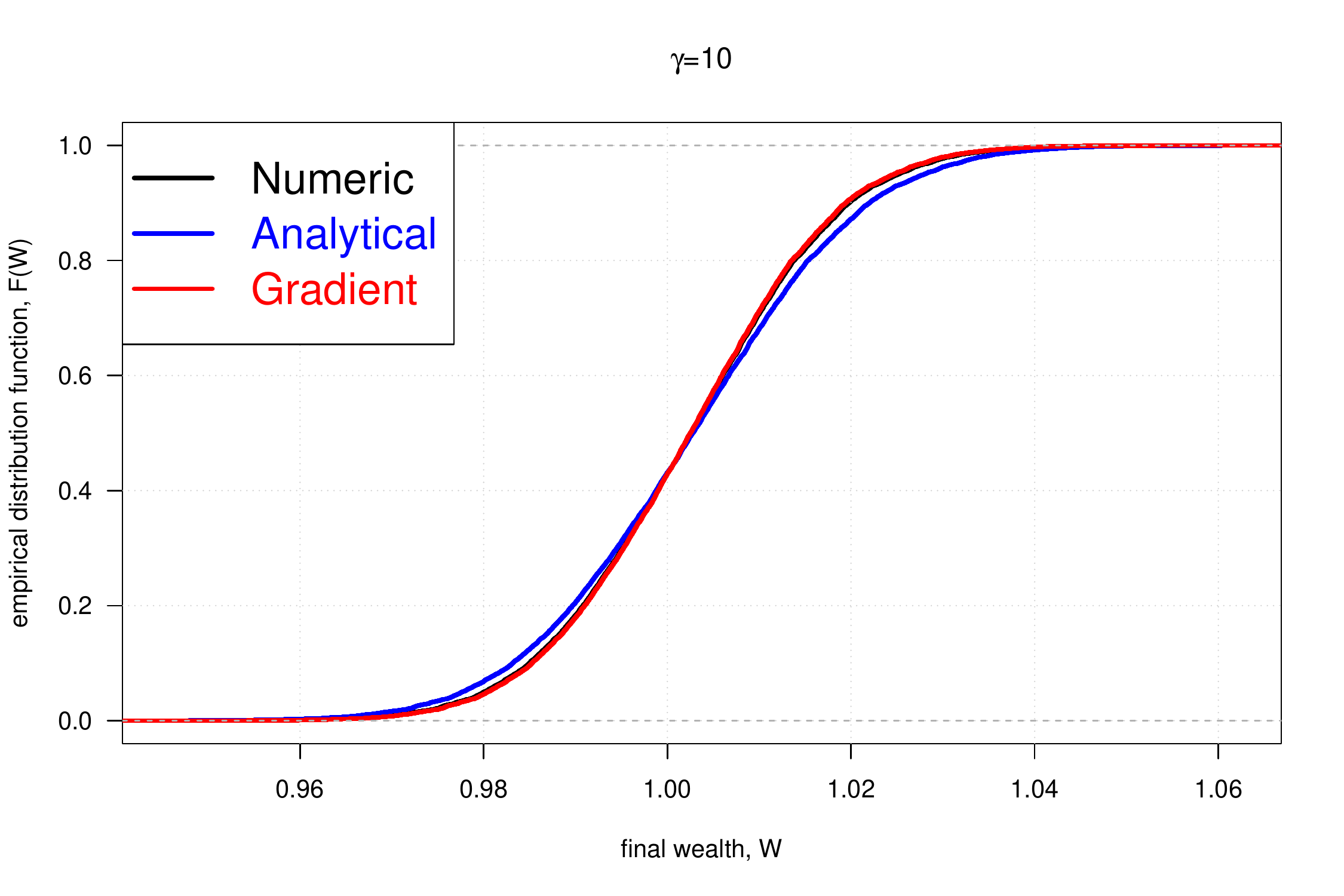}
		\includegraphics[width=0.5\textwidth]{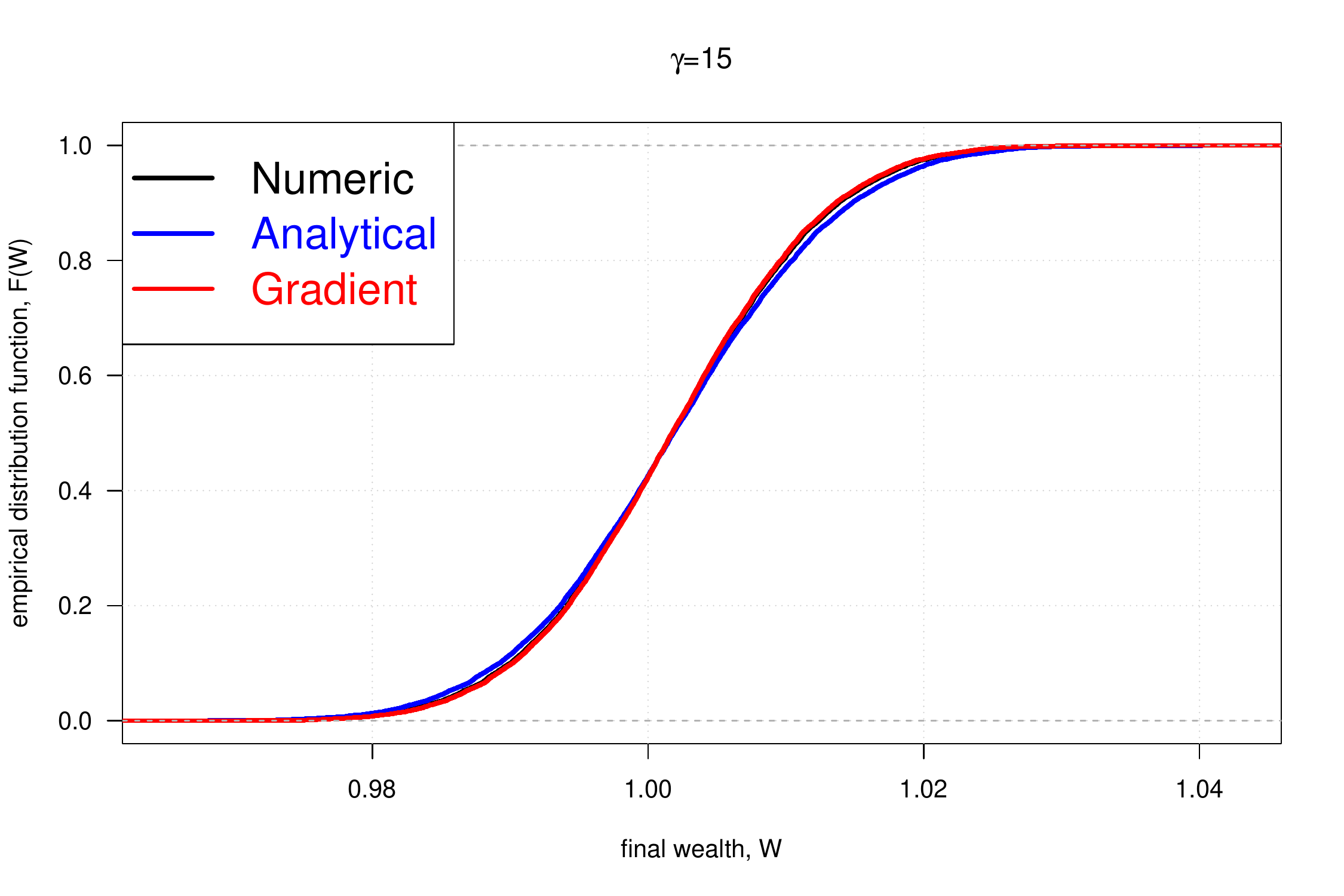}\includegraphics[width=0.5\textwidth]{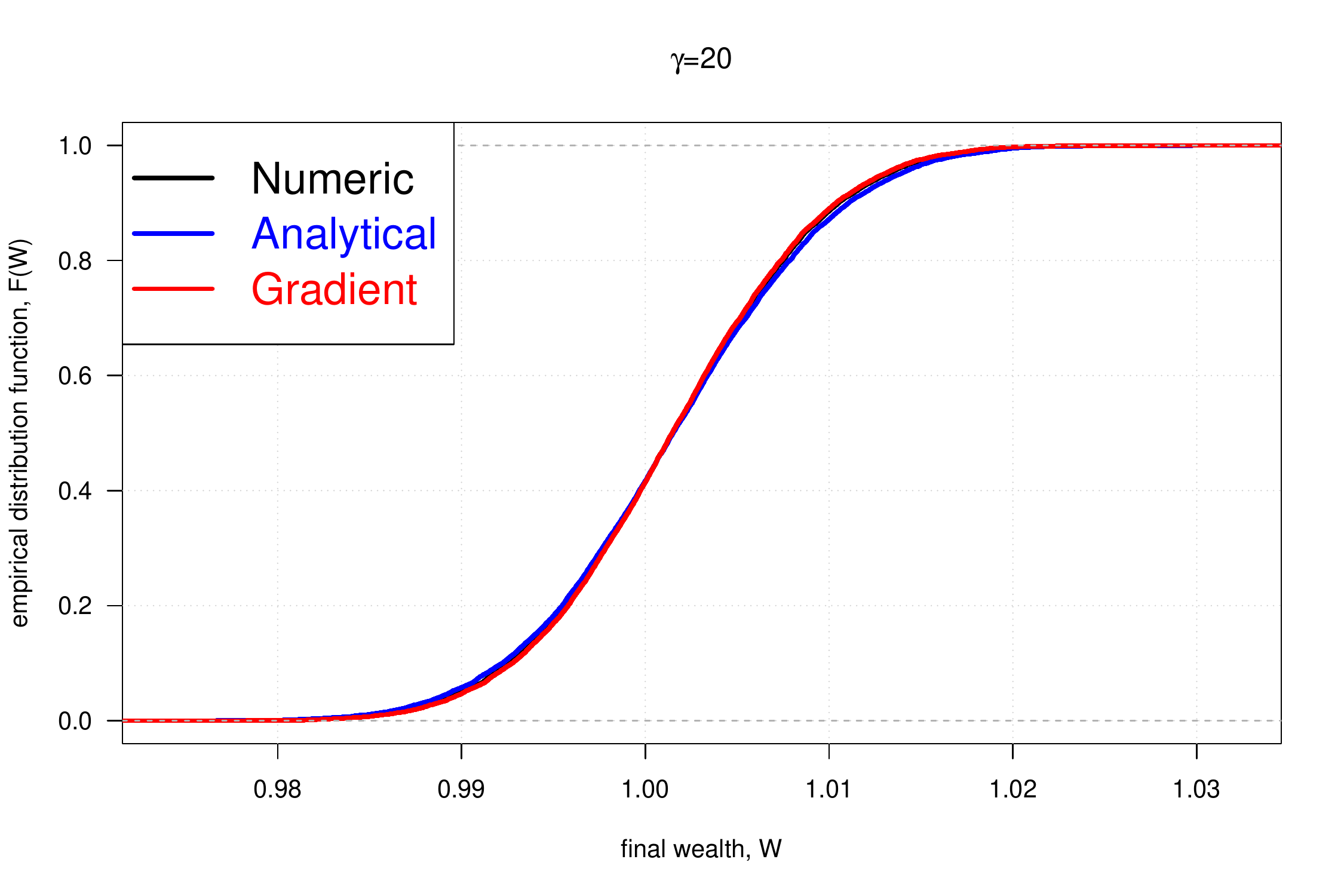}
		\caption{Empirical distribution function of the final wealth for three strategies in case of $\gamma\in \{5, 10, 15, 20\}$}
		\label{fig: ecdf w}
	\end{figure}
	\begin{figure}[]		
		\includegraphics[width=0.5\textwidth]{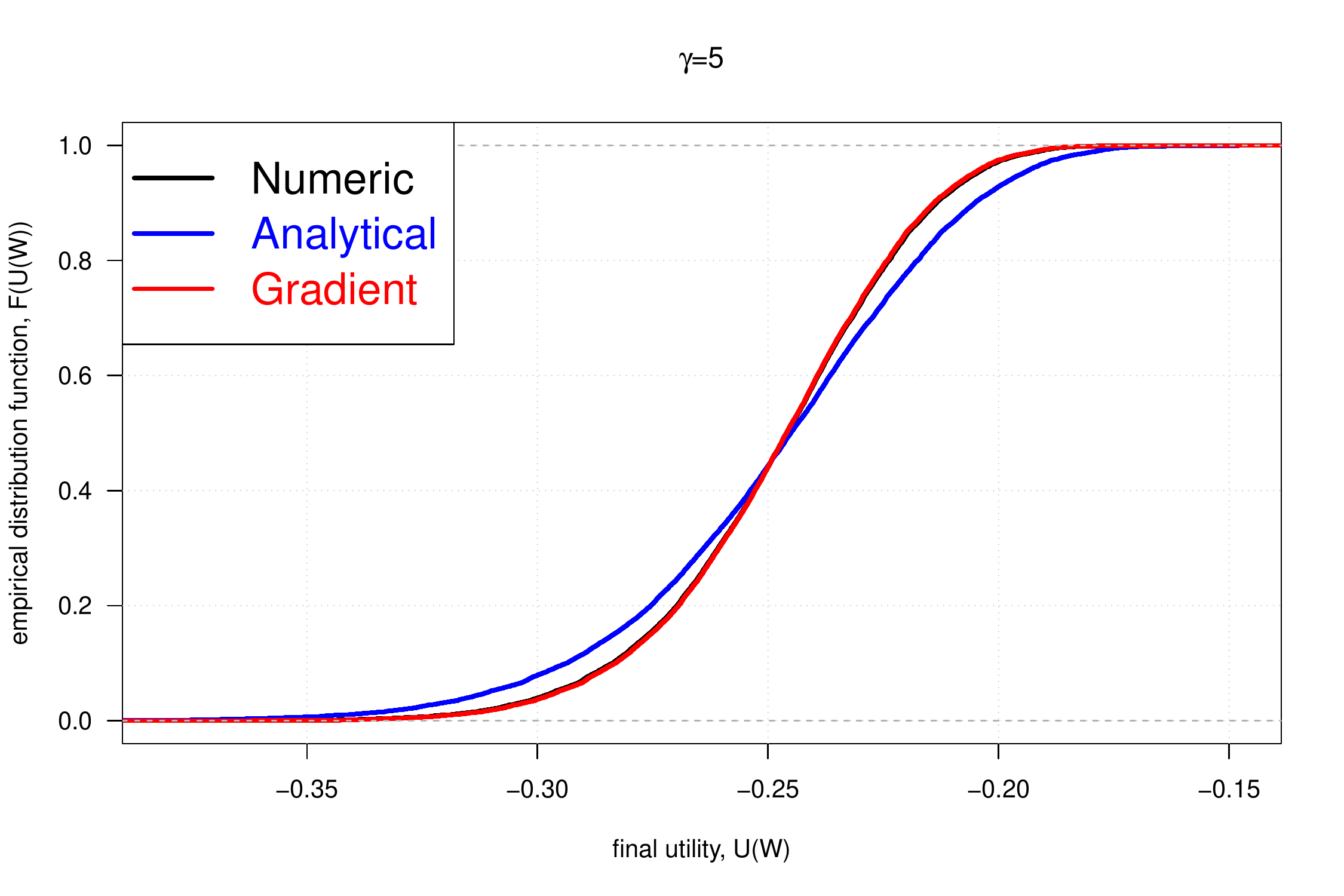}\includegraphics[width=0.5\textwidth]{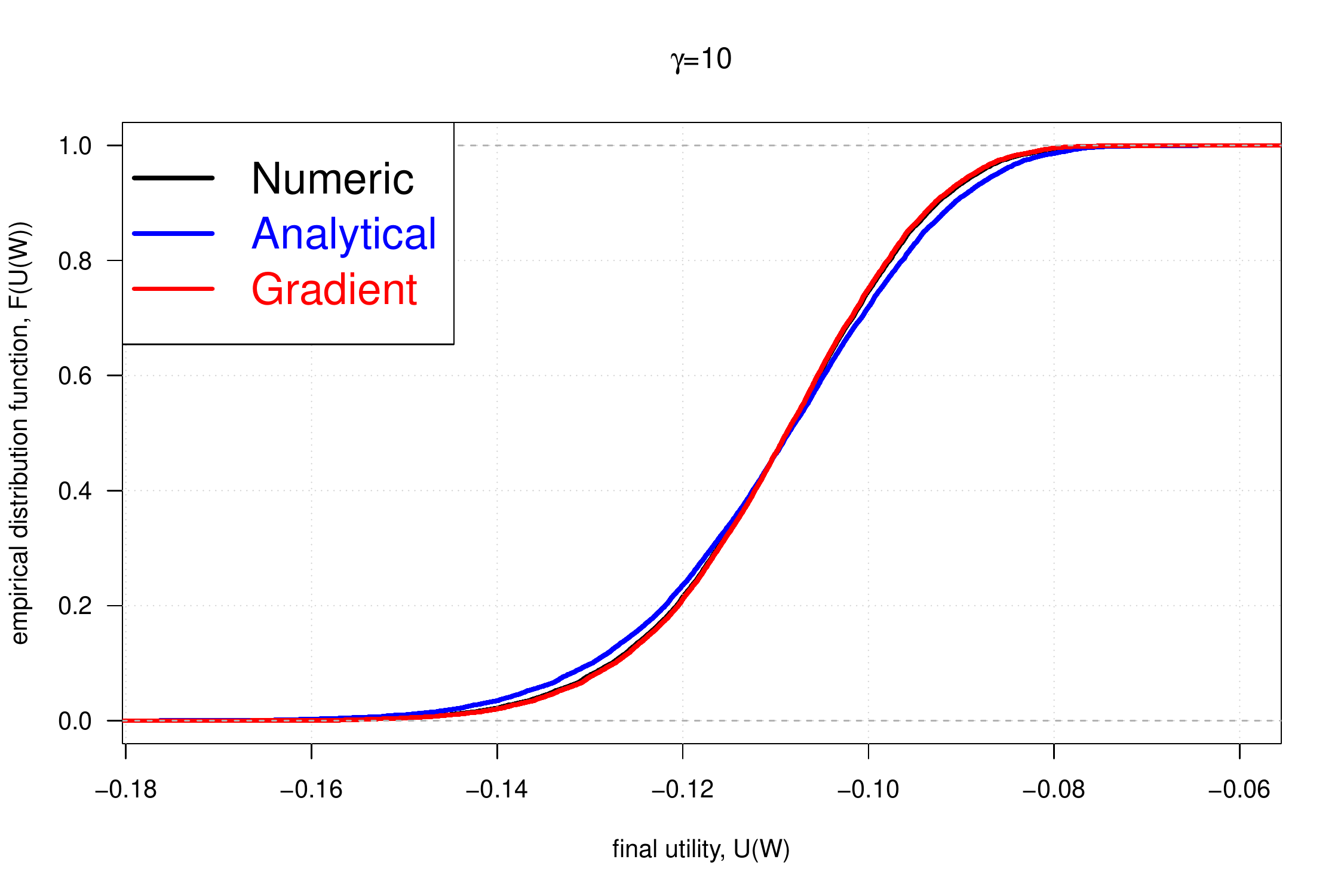}
		\includegraphics[width=0.5\textwidth]{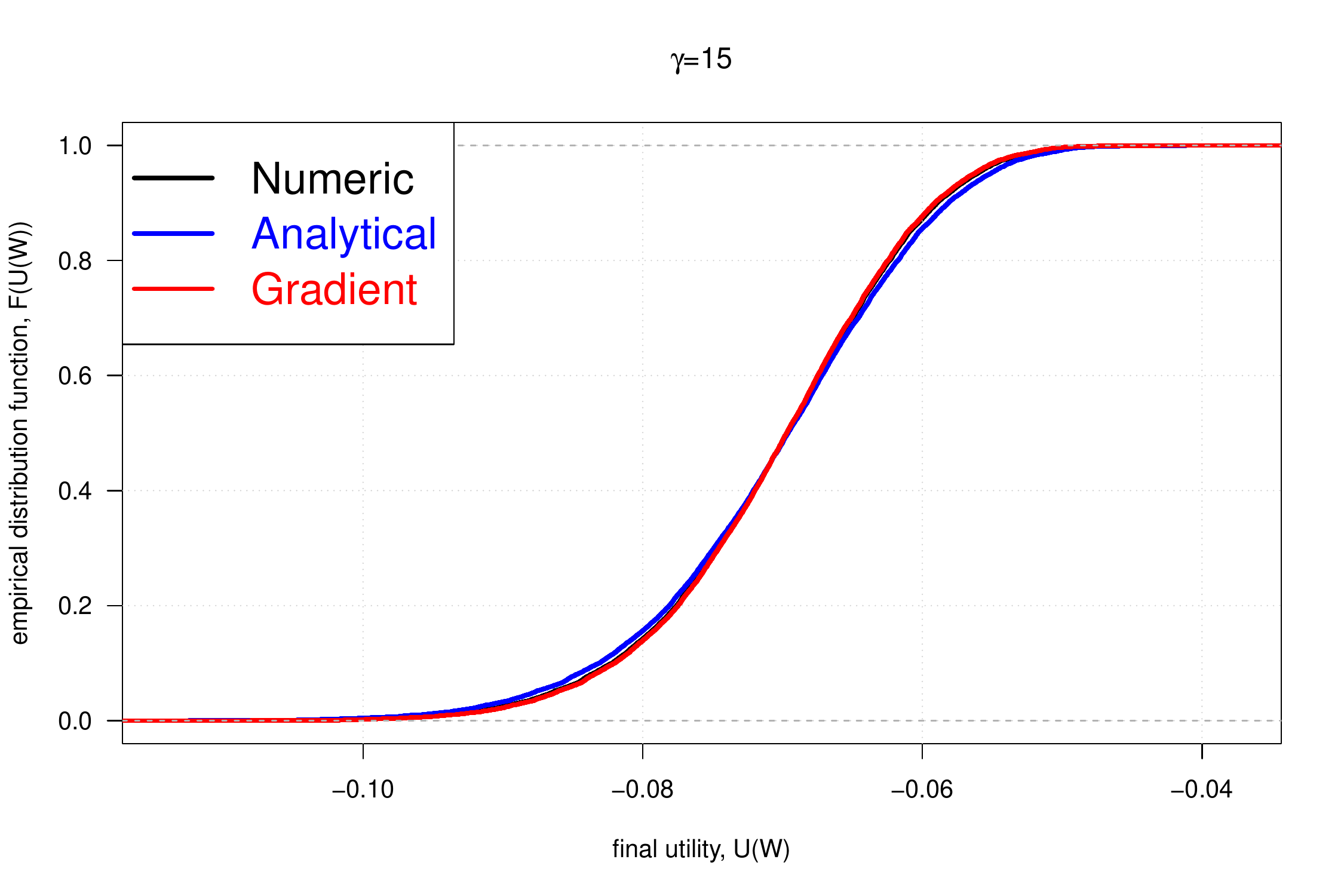}\includegraphics[width=0.5\textwidth]{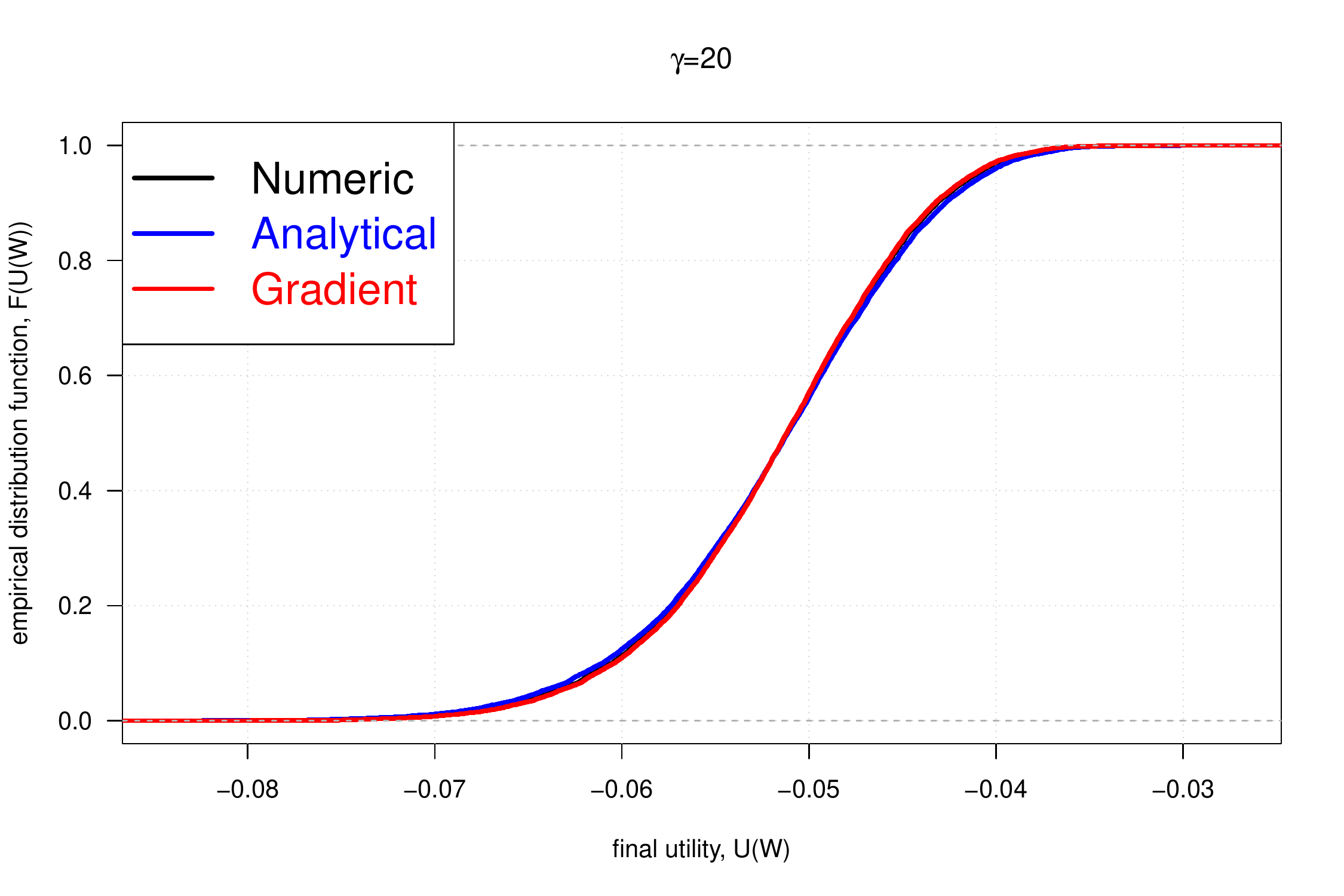}
		\caption{Empirical distribution function of the utility of final wealth for three strategies in case of $\gamma\in \{5, 10, 15, 20\}$}
		\label{fig: ecdf u}
	\end{figure}
	\begin{figure}[]
		\includegraphics[width=\textwidth]{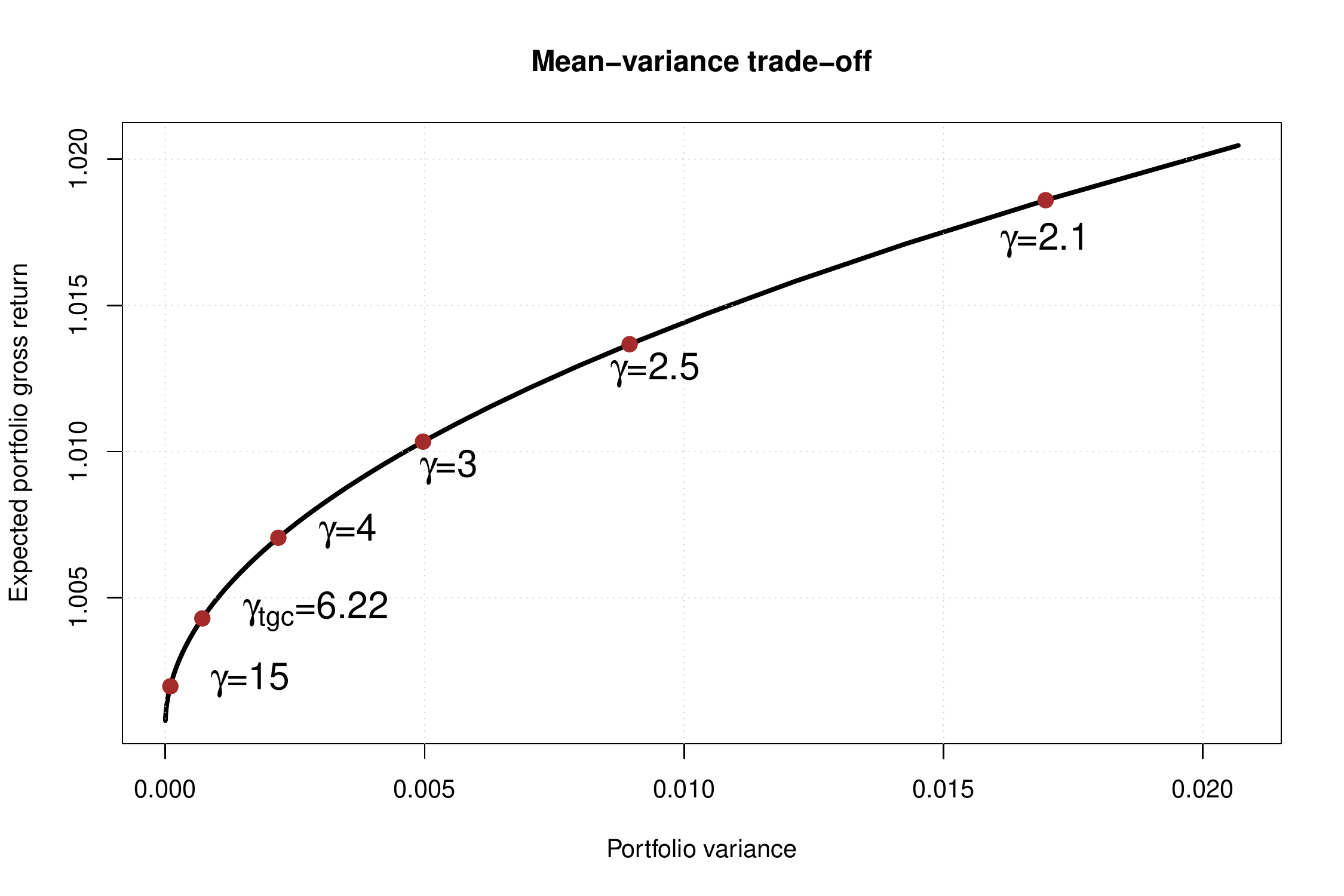}
		\caption{Mean-variance trade-off of the portfolio derived in Theorem \ref{th: optimal portfolio weights theorem}}
		\label{fig: mean-variance}
	\end{figure}

	\begin{figure}[]
	\includegraphics[width=\textwidth]{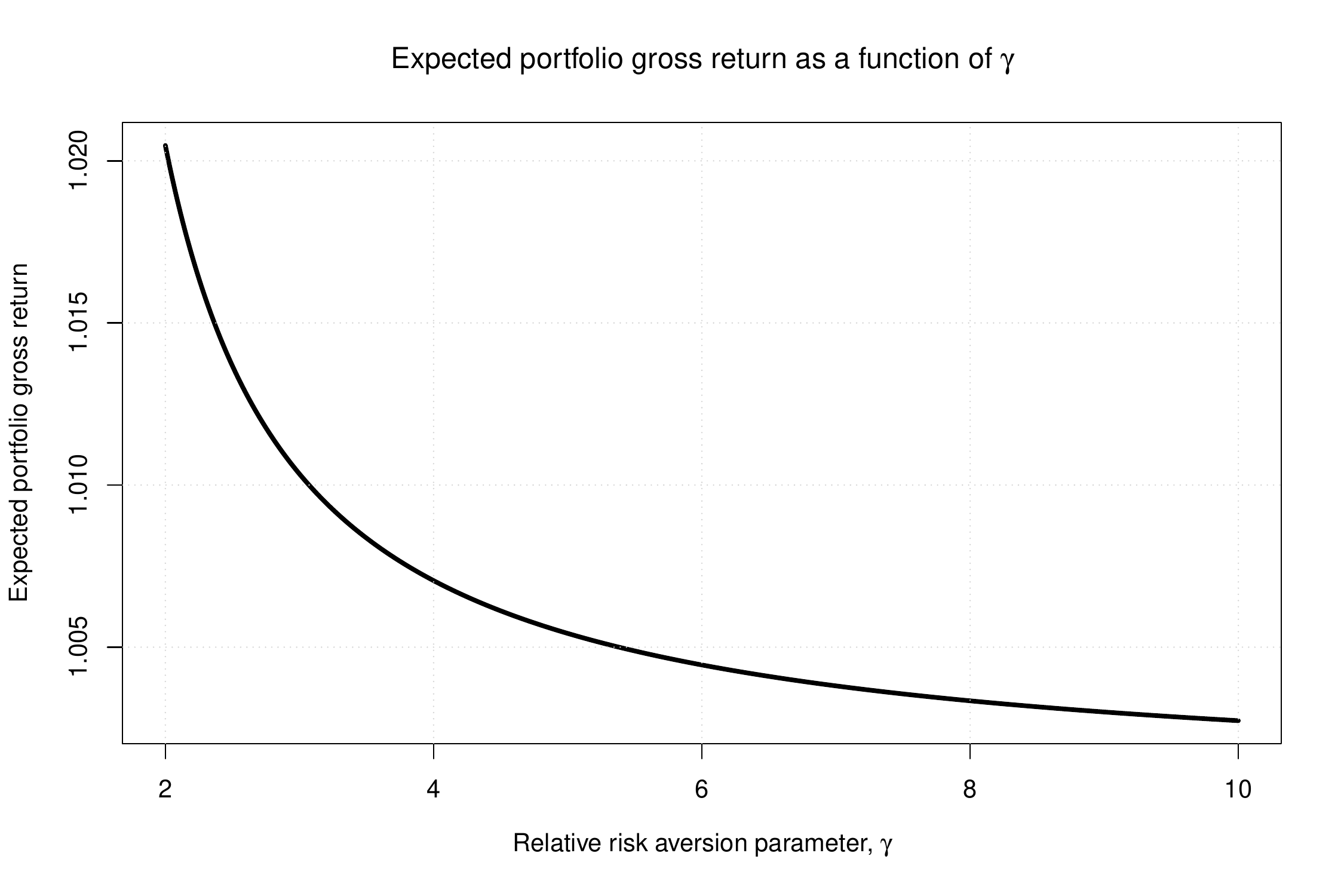}
	\caption{Expected portfolio excess return as a function of $\gamma$}
	\label{fig: g-mean}
	\end{figure}

	\begin{figure}[]
	\includegraphics[width=\textwidth]{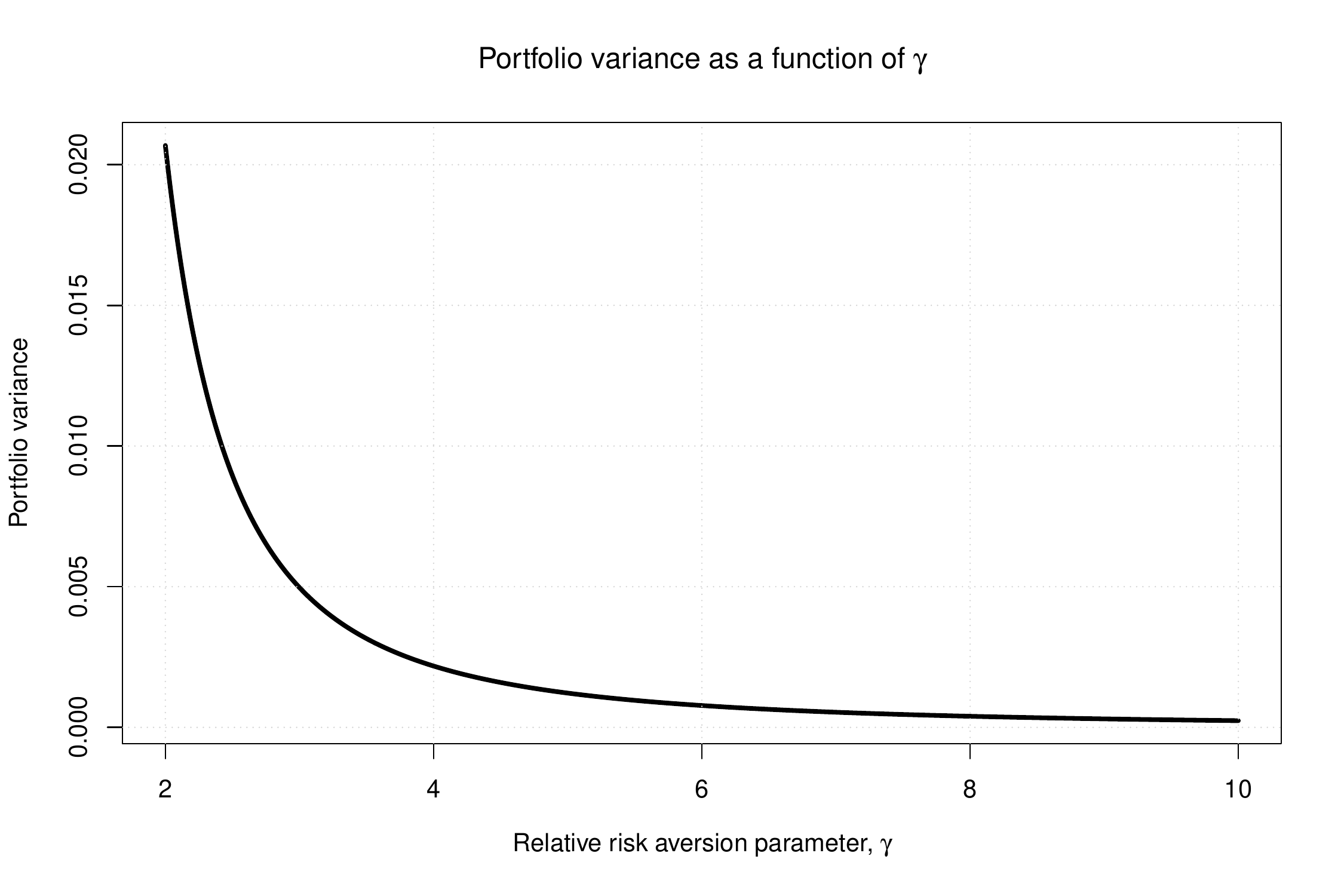}
	\caption{Portfolio variance as a function of $\gamma$}
	\label{fig: g-variance}
	\end{figure}

Table \ref{table: expected power utility} presents the comparison of the strategies. We calculate the sample mean and median of the final utilities for each $\gamma$. In addition, there are sample standard deviation and median absolute deviation in the table. Each column represents three portfolios: analytical, numerical and Gradient Descent.

The first value we are interested in is the sample mean because this is an estimate for the expected value. One can clearly say that the numbers are identical except for $\gamma=5$ and $\gamma=10$ however, the difference is negligible. We observe a similar situation for the median. The values are almost equal, and the difference is minor. But one should point out that for each $\gamma$ analytical solution has a higher median. So one can conclude that the analytical solution slightly dominates but at the same time produces higher variation, which influences the sample mean. That is what we see at sample standard deviation rows. Here the analytical solution has a higher deviation, whereas the Gradient Descent provides the safest portfolio. That also holds for the median absolute deviation. 

To reveal more information from the data, we built the empirical cumulative distribution function for the final wealth and the cumulative distribution function for the utility of wealth. As earlier, for each value $\gamma$, we combine three portfolios. Four plots represent the empirical distribution function of the final wealth In Figure \ref{fig: ecdf w}. One can observe the deviation of the analytical result from two others which approves the higher standard deviation. At the same time, the numerical curve is located directly underneath the one which corresponds to the Gradient Descent solution. But they start to coincide and overlap with an increase of the risk aversion. The identical situation takes place in Figure \ref{fig: ecdf u}, where have the empirical cumulative distribution function for the utility of final wealth.  

As a result, the comparison study shows that the derived analytical solution and the Gradient Descent portfolio provide competitive results compared to the benchmark numerical solution. 

Along with the comparison study, we demonstrate some properties of the weights derived in Theorem \ref{th: optimal portfolio weights theorem} and illustrate the ideas described in Corollary \ref{cor: mean-variance parabola} and Corollary \ref{cor: mean and variance decrease of gama}.
In Figure \ref{fig: mean-variance}, one can find the parabola plot of the expected portfolio gross returns with the corresponding variances. Also, we added several points representing values of the coefficient of relative risk aversion, including the tangency portfolio, derived in Corollary \ref{cor: tangency portfolio}. One can observe that those dots move towards the edge of the parabola as $\gamma$ increases. This behaviour of the portfolio mean and variance being a decrease function of the coefficient of relative risk aversion is also well described in Figure \ref{fig: g-mean} and Figure \ref{fig: g-variance}. 
	
The last part of the simulation study is related to the lower bound of $\gamma$ in Theorem \ref{th: optimal portfolio weights theorem}. Given the parameters $\bmu$ and $\bSigma$, we have $1+4\bmu\T\bSigma\inv\bmu=1.076384$. Meanwhile, the literature provides the median value of relative risk aversion being around seven \citet{barsky1997preference,pennacchi2008theory}. It brings the idea that the lower bound of $1+4\bmu\T\bSigma\inv\bmu$ is not that restrictive from a practical point of view.
	
\section{Summary}
	
	In this paper, we developed two different solutions for the single-period power utility optimization problem. The first one is an approximate result that has analytical formula considering that the asset returns follow a multivariate normal distribution. The resulting portfolio follows a set of properties that provide its efficiency. For instance, the traditional mean-variance trade-off holds for the developed solution - namely: one can earn the higher expected return only by investing in a riskier portfolio. Besides, the expected return and variance of the portfolio are decrease functions of the risk aversion parameter $\gamma$, which is supposed to be true from an economic side \cite[see][]{brandt2009portfolio,pennacchi2008theory,campbell2002strategic}.     
	
	The other solution for optimal weights we derive with the help of the Gradient Descent method. Simply having the historical data allows calculating the vector of optimal weights, which will maximize the expected utility with the predefined precision. It is a universal approach applicable for other utility functions.   
	
	To present the relevance of the derived results in the simulation study, we compare the aforementioned portfolios to the one derived by the Taylor series expansion of the utility function, which is the benchmark numerical solution. We simulate the hypothetical returns from a multivariate normal distribution with the sample mean vector and the variance-covariance matrix estimated from the stock data. 
	Next, for each portfolio and some values of the risk aversion, we calculate the final utility of wealth from the simulated data.
	As a result, all three portfolios have an almost identical sample mean and median for the final utility. That shows the implacability of both solutions derived in the paper. 
	
	\section{Appendix}
		
	\begin{proof}[Proof of Theorem \ref{th: optimal portfolio weights theorem}]
		For the power utility function, one has	
		\begin{equation*}
			E\left[U(W )\right]=E\left[\frac{W_0^{1-\gamma}\left(\Rf+\w\T\bR\right)^{1-\gamma}}{1-\gamma}\right] =  \frac{{W_0^{1-\gamma}}}{1-\gamma} E\left[\left(\Rf+\w\T\bR\right)^{1-\gamma}\right].
		\end{equation*}
		After we consider that $\bR\sim N\left(\bmu,\bSigma\right)$ one can apply the approximation from \citet{bodnar2020mean} and use log-normal distribution instead:
		\begin{equation}\label{eq: objective function}
			E\left[U(W)\right]\approx\frac{{W_0^{1-\gamma}}}{1-\gamma}\exp\left[(1-\gamma)\ln\left(\Rf+\w\T\bmu\right)+\frac{(1-\gamma)^2}{2}\frac{\w\T\bSigma\w}{\left(\Rf+\w\T\bmu\right)^2}\right].
		\end{equation}
		In order to maximize the expected utility \eqref{eq: objective function} with respect to $\w$ it is enough to maximize the next function:
		
		\begin{equation}\label{eq: G}
			G(\w)=\ln\left(\Rf+\w\T{\bmu}\right)+\frac{1-\gamma}{2}\frac{\w\T\bSigma\w}{\left(\Rf+\w\T\bmu\right)^2}.
		\end{equation}
		So, we need to find $\w$ such that 
		\begin{equation}\label{eq: derivative G}
			\frac{\partial G(\w)}{\partial\w}=\textbf{0}.
		\end{equation}

		The partial derivation of $G(\w)$ leads to:
		\begin{align}\label{eq: lagrange derivetive}
			&\frac{\bmu } {\w\T {\bmu + \Rf} }+(1-\gamma)\frac{\bSigma\w\left(\Rf+\w\T\bmu\right)^2-\w\T\bSigma\w(\Rf+\w\T {\bmu} ){\bmu} }{(\Rf+\w\T {\bmu} )^4}=\mathbf{0},\nonumber
			\\\text{or}\nonumber\\
			&\bmu+(1-\gamma)\frac{\bSigma\w\left(\Rf+\w\T\bmu\right)-\w\T\bSigma\w\bmu}{(\Rf+\w\T {\bmu} )^2}=\mathbf{0}.	
		\end{align}
		
		The expression above implies 
		\begin{equation}\label{eq: weights with c}
			\w^*=c\bSigma\inv\bmu,
		\end{equation}
		with $c \in \mathbb{R}$ and consequently we have
		$$\left(1+(1-\gamma)\frac{c(\Rf+cJ)-c^2J}{(\Rf+cJ)^2}\right)\bmu=\textbf{0},$$
		what leads to
	
	\begin{equation}\label{eq: foc equation}
		(\Rf+cJ)^2+(1-\gamma)c\Rf=0,
	\end{equation}
		where
		\begin{equation}\label{J}
			J=\bmu\T\bSigma\inv\bmu.
		\end{equation}
	Let 
	$$D=\frac{(\gamma - 1)^2}{4} - (\gamma - 1)J.$$
		Solving the equation \eqref{eq: foc equation} with respect to $c$ we get

		\begin{equation*}
	c_{\pm}=\Rf\frac{\frac{\gamma - 1}{2} - J \pm \sqrt{D}}{J^2}
		\end{equation*}		
				And we have
				
		$$\w_{+}=\Rf\frac{\frac{\gamma - 1}{2} - J + \sqrt{D}}{J^2}\bSigma\inv\bmu,$$
		$$\w_{-}=\Rf\frac{\frac{\gamma - 1}{2} - J - \sqrt{D}}{J^2}\bSigma\inv\bmu.$$
				
		Thus, we have the solution if $D \ge 0$ or $\gamma \in (0,1)\cup[1+4J,\infty)$. 
		
		Since the solution of first-order conditions \eqref{eq: derivative G} only provide the critical points, it is necessary to discover which one gives a higher value of the objective function \eqref{eq: G}.

\begin{equation}\label{eq: argument difference}
	\begin{split}
		&G(\w_{+})-G(\w_{-})=\ln\frac{\Rf+c_{+}J}{\Rf+c_{-}J}+\frac{1-\gamma}{2}\frac{c^2_{+}J}{(\gamma-1)c_{+}\Rf}-\frac{1-\gamma}{2}\frac{c_{-}^2J}{(\gamma-1)c_{-}\Rf}\\
		&=\ln\frac{\gamma-1+2\sqrt{\D}}{\gamma-1-2\sqrt{\D}}-\frac{c_{+}J}{2\Rf}+\frac{c_{-}J}{2\Rf}=\ln\frac{\gamma-1+2\sqrt{\D}}{\gamma-1-2\sqrt{\D}}-\frac{\sqrt{\D}}{J}\\
	\end{split}
\end{equation}
Next we show that the expression \eqref{eq: argument difference} is a decrease function of $\gamma$ for $\gamma>1+4J$. Besides, it is equal to zero for $\gamma=1+4J$. Together, this implies a positive \eqref{eq: argument difference} for $\gamma>1+4J$. So let us check it:			
\begin{equation}
	\begin{split}
		&\frac{\partial}{\partial\gamma}\left[\ln\frac{\gamma-1+2\sqrt{\D}}{\gamma-1-2\sqrt{\D}}-\frac{\sqrt{\D}}{J}\right] = \frac{\partial}{\partial\gamma}\left[G(\w_{+})-G(\w_{-})\right]\\
		&=\frac{\gamma-1-2\sqrt{\D}}{\gamma-1+2\sqrt{\D}}\frac{\left(1+2\frac{\partial \sqrt{D}}{\partial\gamma}\right)\left(\gamma-1-2\sqrt{\D}\right)-\left(1-2\frac{\partial \sqrt{D}}{\partial\gamma}\right)\left(\gamma-1+2\sqrt{\D}\right)}{\left(\gamma-1-2\sqrt{\D}\right)^2}-\frac{\frac{\partial \sqrt{D}}{\partial\gamma}}{J}\\
		&=\frac{4\frac{\partial \sqrt{D}}{\partial\gamma}(\gamma-1)-4\sqrt{\D}}{4(\gamma-1)J}-\frac{\frac{\partial \sqrt{D}}{\partial\gamma}}{J}=\frac{\sqrt{\D}}{J(1-\gamma)}<0,\ \text{if}\ \gamma>1+4J .\\
	\end{split}
\end{equation}
Meanwhile
\begin{equation}
	\begin{split}
		&\left[G(\w_{+})-G(\w_{-})\right]_{\gamma=1+4J}=\ln\frac{4J}{4J}=0.
	\end{split}
\end{equation}
Thus, the expression \eqref{eq: argument difference} is a decrease function of $\gamma$ and negative for $\gamma>1+4J$. 
Hence, fro all $\gamma>1+4J$ the maximum of \eqref{eq: objective function} is attained at 

$$\w=\Rf\frac{\frac{\gamma - 1}{2} - J - \sqrt{D}}{J^2}\bSigma\inv\bmu.$$

		It is notable that $\Rf+\w\T\bmu$ must be positive, but for $0<\gamma<1$, $\Rf+\w_{-}\T\bmu$ becomes negative and, consequently, $\w_{+}$ is an only candidate for an extrema point. In order to investigate the type of the extrema we calculate the second derivative of the objective function and analyse its sign:
		
		\begin{equation}
		\begin{split}
			&\frac{\partial^2}{\partial\w^2}\left[\ln\left(\Rf+\w\T{\bmu}\right)+\frac{1-\gamma}{2}\frac{\w\T\bSigma\w}{\left(\Rf+\w\T\bmu\right)^2}\right]\\
			&\overset{\eqref{eq: lagrange derivetive}}=\frac{\partial}{\partial\w\T}\left[\frac{\bmu } {\w\T {\bmu + \Rf} }+(1-\gamma)\frac{\bSigma\w\left(\Rf+\w\T\bmu\right)^2-\w\T\bSigma\w(\Rf+\w\T {\bmu} ){\bmu} }{(\Rf+\w\T {\bmu} )^4}\right]\\
			&=\frac{\partial}{\partial\w\T}\left[\frac{\bmu } {\w\T {\bmu + \Rf} } + (1-\gamma)\frac{\bSigma\w}{\left(\Rf+\w\T\bmu\right)^2}-(1-\gamma)\frac{\w\T\bSigma\w\bmu}{(\Rf+\w\T {\bmu} )^3}\right]\\
			&=-\frac{\bmu\bmu\T}{\left(\Rf+\w\T\bmu\right)^2}+(1-\gamma)\frac{\bSigma\left(\Rf+\w\T\bmu\right)^2-2\bSigma\w\bmu\T\left(\Rf+\w\T\bmu\right)}{\left(\Rf+\w\T\bmu\right)^4}\\
			&- (1-\gamma)\frac{2\bmu\w\T\bSigma(\Rf+\w\T {\bmu} )^3-3(\Rf+\w\T {\bmu} )^2\w\T\bSigma\w\bmu\bmu\T}{(\Rf+\w\T {\bmu} )^6}\\
			&\overset{\eqref{eq: weights with c}}=-\frac{\bmu\bmu\T}{\left(\Rf+cJ\right)^2}+(1-\gamma)\frac{\bSigma\left(\Rf+cJ\right)^2-2c\bmu\bmu\T\left(\Rf+cJ\right)}{\left(\Rf+cJ\right)^4}\\
			&- (1-\gamma)\frac{2c\bmu\bmu\T(\Rf+cJ )^3-3(\Rf+cJ )^2c^2J\bmu\bmu\T}{(\Rf+cJ )^6}\\
			&=-\frac{\bmu\bmu\T}{\left(\Rf+cJ\right)^2}+(1-\gamma)\frac{\bSigma}{\left(\Rf+cJ\right)^2}-(1-\gamma)\frac{4c\bmu\bmu\T}{\left(\Rf+cJ\right)^3}+ (1-\gamma)\frac{3c^2J\bmu\bmu\T}{(\Rf+cJ )^4}\\
			&=(1-\gamma)\frac{\bSigma}{\left(\Rf+cJ\right)^2} + \frac{\bmu\bmu\T}{(\Rf+cJ)^4}\left(-{\left(\Rf+cJ\right)^2}-(1-\gamma){4c}{\left(\Rf+cJ\right)}+ (1-\gamma){3c^2J}\right)\\	
			&=(1-\gamma)\frac{\bSigma}{\left(\Rf+cJ\right)^2} + \frac{\bmu\bmu\T}{(\Rf+cJ)^4}\left((1-\gamma)c\Rf-(1-\gamma){4c}{\left(\Rf+cJ\right)}+ (1-\gamma){3c^2J}\right)\\
			&=(1-\gamma)\frac{\bSigma}{\left(\Rf+cJ\right)^2} + (1-\gamma)c\frac{\bmu\bmu\T}{(\Rf+cJ)^4}\left(-3\Rf-  {cJ}\right)\\
			&\overset{\eqref{eq: foc equation}}=(1-\gamma)\frac{\bSigma}{\left(\Rf+cJ\right)^2} + \frac{\bmu\bmu\T}{\Rf(\Rf+cJ)^2}\left(3\Rf +  {cJ}\right)\\			
	\end{split}
	\end{equation}
		Thus, the second derivative is positive for $\gamma\in\left(0,1\right)$ and in this case $\w_{+}$ provide a local minima of the function $G(\w)$.

		Consequently the	optimal portfolio wights can be calculated for $\gamma\geq 1+4\bmu\T\bSigma\inv\bmu$ with the formula
		\begin{equation}\label{eq: single period approximate optimal weights}
			\w^*=\Rf\frac{\frac{\gamma - 1}{2} - \bmu\T\bSigma\inv\bmu - \sqrt{\frac{(\gamma - 1)^2}{4} - (\gamma - 1) \bmu\T\bSigma\inv\bmu}}{(\bmu\T\bSigma\inv\bmu)^2}\bSigma\inv\bmu.
		\end{equation} 
	\end{proof}
	
	\begin{proof}[Proof of the Corollary \ref{cor: mean-variance parabola}]
	From the proof of Theorem \ref{th: optimal portfolio weights theorem} we have:
	
	 $\w=c\bSigma\inv\bmu \Rightarrow \w\T\bmu = cJ\ \text{and}\  \w\T\bSigma\w=c^2J\Rightarrow (\w\T\bmu)^2=\w\T\bSigma\w J$.
	\end{proof}

	\begin{proof}[Proof of the Corollary \ref{cor: mean and variance decrease of gama}]
		\begin{equation*}
			\begin{split}
				\frac{\partial}{\partial\gamma}(\Rf + \w^{*\prime}\bmu)&=\frac{\partial}{\partial\gamma}\left[\Rf\frac{\gamma-1  - \sqrt{(\gamma-1)^2-4(\gamma-1)J}}{2J}\right]\\
				&=\frac{\Rf}{2J}\left(1-\frac{2(\gamma-1)-4J}{2\sqrt{(\gamma-1)^2-4(\gamma-1)J}}\right)\\
				&=\frac{\Rf}{2J}\frac{\sqrt{(\gamma-1)^2-4(\gamma-1)J}- (\gamma-1)+2J}{\sqrt{(\gamma-1)^2-4(\gamma-1)J}}\\
				&=\frac{\Rf}{2J\sqrt{\D}}\frac{{(\gamma-1)^2-4(\gamma-1)J}- (\gamma-1-2J)^2}{\sqrt{(\gamma-1)^2-4(\gamma-1)J}+ \gamma-1-2J}\\
				&=\frac{\Rf}{2J\sqrt{\D}}\frac{ - 4J^2}{\sqrt{(\gamma-1)^2-4(\gamma-1)J}+ \gamma-1-2J}\overset{\gamma\geq 1+4J}<0.
			\end{split}
		\end{equation*}
	Taking into account that $\boldsymbol{\omega}^{*\prime}\bmu + \Rf$ is positive for $\gamma>1+4J$:
			\begin{equation*}
		\begin{split}
			\Rf+\boldsymbol{\omega}^{*\prime}\bmu&=\Rf\frac{\gamma-1  - \sqrt{(\gamma-1)^2-4(\gamma-1)J}}{2J}-\Rf\\
			&=\Rf\frac{\gamma-1-2J-\sqrt{(\gamma-1)^2-4(\gamma-1)J}}{2J}\\
			&=\frac{\Rf}{2J}\frac{(\gamma-1-2J)^2-(\gamma-1)^2+4(\gamma-1)J}{\gamma-1-2J+\sqrt{(\gamma-1)^2-4(\gamma-1)J}}\\
			&=\frac{\Rf}{2J}\frac{ 4J^2}{\gamma-1-2J+\sqrt{(\gamma-1)^2-4(\gamma-1)J}}\overset{\gamma\geq1+4J}>0,
		\end{split}
	\end{equation*}
we have that portfolio variance also decreases on $\gamma$ as a square of $\boldsymbol{\omega}^{*\prime}\bmu$.
	\end{proof}
	
		\begin{proof}[Proof of the Corollary \ref{cor: tangency portfolio}]
If all the wealth is split between risky assets, then sum of weights $\w$ is equal to one. Consequently, from equation \eqref{eq: weights with c} we receive $1=\1\T\w=c\1\bSigma\inv\bmu$ and $c=\frac{1}{\1\T\bSigma\inv\bmu}$. 
One can also find the corresponding value of the risk aversion $\gamma$ form equation \eqref{eq: foc equation}:
$$\left(\Rf+\frac{J}{\1\T\bSigma\inv\bmu}\right)^2+(1-\gamma)\frac{\Rf}{\1\T\bSigma\inv\bmu}=0 \Rightarrow \gamma_{tgc}=\left(\Rf+\frac{J}{\1\T\bSigma\inv\bmu}\right)^2\frac{\1\T\bSigma\inv\bmu}{\Rf}+1.$$
	
It is also notable that if $\1\T\bSigma\inv\bmu$ is positive then because of $(a+b)^2\geq4ab$ it holds:

$$\gamma_{tgc}=\left(\Rf+\frac{J}{\1\T\bSigma\inv\bmu}\right)^2\frac{\1\T\bSigma\inv\bmu}{\Rf}+1\geq 4\Rf\frac{J}{\1\T\bSigma\inv\bmu}\frac{\1\T\bSigma\inv\bmu}{\Rf}+1= 1 + 4J$$
\end{proof}

	\bibliography{2}{}

\begin{thebibliography}{}

\bibitem[\protect\astroncite{Andersen et~al.}{2000}]{andersen2000exchange}
Andersen, T.~G., Bollerslev, T., Diebold, F.~X., and Labys, P. (2000).
\newblock Exchange rate returns standardized by realized volatility are
  (nearly) gaussian.
\newblock {\em NBER Working Paper}, (w7488).

\bibitem[\protect\astroncite{Barberis}{2000}]{barberis2000investing}
Barberis, N. (2000).
\newblock Investing for the long run when returns are predictable.
\newblock {\em The Journal of Finance}, 55(1):225--264.

\bibitem[\protect\astroncite{Barsky et~al.}{1997}]{barsky1997preference}
Barsky, R.~B., Juster, F.~T., Kimball, M.~S., and Shapiro, M.~D. (1997).
\newblock Preference parameters and behavioral heterogeneity: An experimental
  approach in the health and retirement study.
\newblock {\em The Quarterly Journal of Economics}, 112(2):537--579.

\bibitem[\protect\astroncite{Bodnar et~al.}{2020}]{bodnar2020mean}
Bodnar, T., Ivasiuk, D., Parolya, N., and Schmid, W. (2020).
\newblock Mean-variance efficiency of optimal power and logarithmic utility
  portfolios.
\newblock {\em Mathematics and Financial Economics}, 14:675--698.

\bibitem[\protect\astroncite{Bodnar et~al.}{2015a}]{bodnar2015closed}
Bodnar, T., Parolya, N., and Schmid, W. (2015a).
\newblock A closed-form solution of the multi-period portfolio choice problem
  for a quadratic utility function.
\newblock {\em Annals of Operations Research}, 229(1):121--158.

\bibitem[\protect\astroncite{Bodnar et~al.}{2015b}]{bodnar2015exact}
Bodnar, T., Parolya, N., and Schmid, W. (2015b).
\newblock On the exact solution of the multi-period portfolio choice problem
  for an exponential utility under return predictability.
\newblock {\em European Journal of Operational Research}, 246(2):528--542.

\bibitem[\protect\astroncite{Brandt}{2009}]{brandt2009portfolio}
Brandt, M. (2009).
\newblock Portfolio choice problems.
\newblock {\em Handbook of financial econometrics}, 1:269--336.

\bibitem[\protect\astroncite{Brandt et~al.}{2005}]{brandt2005simulation}
Brandt, M.~W., Goyal, A., Santa-Clara, P., and Stroud, J.~R. (2005).
\newblock A simulation approach to dynamic portfolio choice with an application
  to learning about return predictability.
\newblock {\em The Review of Financial Studies}, 18(3):831--873.

\bibitem[\protect\astroncite{Brandt and Santa-Clara}{2006}]{brandt2006dynamic}
Brandt, M.~W. and Santa-Clara, P. (2006).
\newblock Dynamic portfolio selection by augmenting the asset space.
\newblock {\em The Journal of Finance}, 61(5):2187--2217.

\bibitem[\protect\astroncite{Broadie and Shen}{2017}]{broadie2017numerical}
Broadie, M. and Shen, W. (2017).
\newblock Numerical solutions to dynamic portfolio problems with upper bounds.
\newblock {\em Computational Management Science}, 14(2):215--227.

\bibitem[\protect\astroncite{Campbell and
  Viceira}{2002}]{campbell2002strategic}
Campbell, J.~Y. and Viceira, L.~M. (2002).
\newblock {\em Strategic asset allocation: portfolio choice for long-term
  investors}.
\newblock Clarendon Lectures in Economic.

\bibitem[\protect\astroncite{{\c{C}}anako{\u{g}}lu and
  {\"O}zekici}{2010}]{ccanakouglu2010portfolio}
{\c{C}}anako{\u{g}}lu, E. and {\"O}zekici, S. (2010).
\newblock Portfolio selection in stochastic markets with hara utility
  functions.
\newblock {\em European Journal of Operational Research}, 201(2):520--536.

\bibitem[\protect\astroncite{{\c{C}}anako{\u{g}}lu and
  {\"O}zekici}{2012}]{ccanakouglu2012hara}
{\c{C}}anako{\u{g}}lu, E. and {\"O}zekici, S. (2012).
\newblock Hara frontiers of optimal portfolios in stochastic markets.
\newblock {\em European Journal of Operational Research}, 221(1):129--137.

\bibitem[\protect\astroncite{Friedman et~al.}{2001}]{friedman2001elements}
Friedman, J., Hastie, T., Tibshirani, R., et~al. (2001).
\newblock {\em The elements of statistical learning}, volume~1.
\newblock Springer series in statistics New York.

\bibitem[\protect\astroncite{G{\'e}ron}{2019}]{geron2019hands}
G{\'e}ron, A. (2019).
\newblock {\em Hands-on machine learning with Scikit-Learn, Keras, and
  TensorFlow: Concepts, tools, and techniques to build intelligent systems}.
\newblock O'Reilly Media.

\bibitem[\protect\astroncite{Goodfellow et~al.}{2016}]{goodfellow-et-al-2016}
Goodfellow, I., Bengio, Y., and Courville, A. (2016).
\newblock {\em Deep Learning}.
\newblock MIT Press.

\bibitem[\protect\astroncite{Merton}{1972}]{merton1972analytic}
Merton, R.~C. (1972).
\newblock An analytic derivation of the efficient portfolio frontier.
\newblock {\em Journal of financial and quantitative analysis}, pages
  1851--1872.

\bibitem[\protect\astroncite{Pennacchi}{2008}]{pennacchi2008theory}
Pennacchi, G.~G. (2008).
\newblock {\em Theory of asset pricing}.
\newblock Pearson/Addison-Wesley Boston.

\end{thebibliography}
	
\end{document}